\documentclass[journal,comsoc]{IEEEtran}

\usepackage[T1]{fontenc} 
\usepackage{setspace}

\usepackage[caption=false]{subfig}
\usepackage[pdftex]{graphicx}

\usepackage{tikz}

\definecolor{gold}{rgb}{0.85,.66,0}

\usepackage{cite}

\usepackage{amsthm}
\usepackage{amssymb}
\usepackage{amsmath}
\usepackage[cmintegrals]{newtxmath}

\usepackage{bm}

\makeatletter
\def\munderbar#1{\underline{\sbox\tw@{$#1$}\dp\tw@\z@\box\tw@}}
\makeatother

\newtheorem{theorem}{Theorem}

\usepackage{soul}
\usepackage[normalem]{ulem} 

\usepackage{algorithm}
\usepackage{algpseudocode}

\ifCLASSINFOpdf
\else
\fi

\newcommand{\eye}[1]{\mathbf{I}_{#1}}

\DeclareMathOperator*{\argmin}{arg\,min}

\newcommand{\ltwonorm}[1]{\lVert#1\rVert^{2}_{2}}

\begin{document}



\newcommand{\papertitle}{
Decentralized Design of Fast Iterative Receivers for Massive and Extreme-Large MIMO Systems
}


\title{\papertitle}

\author{
    {{Victor Croisfelt}},
    {{Taufik Abrão}},
    {{Abolfazl Amiri}},
    {{Elisabeth de Carvalho}}, 
    {{Petar Popovski}}
    \thanks{V. Croisfelt is with the Electrical Engineering Department, Universidade de S\~{a}o Paulo, Escola Politécnica, São Paulo, Brazil; \texttt{victorcroisfelt@usp.br}.}%
    \thanks{{T. Abrão is with the Electrical Engineering Department, State University of Londrina, PR, Brazil.  E-mail: \texttt{taufik@uel.br}}}%
    \thanks{A. Amiri, E. de Carvalho, and P. Popovski are with the Department of Electronic Systems, Technical Faculty of IT and Design; Aalborg University,	Denmark; \texttt{\{aba; edc; petarp\}@es.aau.dk}.}%
}


\maketitle

\begin{abstract} 
Despite the extensive use of a centralized approach to design receivers at the base station for massive multiple-input multiple-output (M-MIMO) systems, their actual implementation is a major challenge due to several bottlenecks imposed by the large number of antennas. One way to deal with this problem is by fully decentralizing the classic zero-forcing receiver across multiple processing nodes based on the gradient descent method. In this paper, we first explicitly relate this decentralized receiver to a distributed version of the Kaczmarz algorithm and to the use of the successive interference cancellation (SIC) philosophy to mitigate the residual across nodes. In addition, we propose two methods to further accelerate the initial convergence of these iterative decentralized receivers by exploring the connection with the Kaczmarz algorithm: 1) a new Bayesian distributed receiver, which can eliminate noise on an iteration basis; 2) a more practical method for choosing the relaxation parameter. The discussion also consider spatial non-stationarities that arise when the antenna arrays are extremely large (XL-MIMO). We were able to improve the numerical results for both spatially stationary and non-stationary channels, but mainly the non-stationary performance can still be improved compared to the centralized ZF receiver. Future research directions are provided with the aim of further improving the applicability of the receiver based on the principle of successive residual cancellation (SRC). 
\end{abstract}


%
\IEEEpeerreviewmaketitle

\section{Introduction}\label{sec:intro}
A key problem in the hardware implementation of massive multiple-input multiple-output (M-MIMO) systems is the centralization of processing tasks at the base station (BS) \cite{Sanchez2020,Shepard2012}. Such an approach is limited by a tremendous amount of information that must be transferred internally from the modules of the large number of antennas to the \emph{centralized processing unit} (CPU). The latter 
often manages all the processing of the baseband signals. A natural solution to this problem is to decentralize the signal reception tasks \cite{Shepard2012}, implementing more nodes called \emph{remote processing units} (RPUs) that are capable of handling received signals as close as possible to the antennas. While this reduces the amount of exchanging information, other challenges emerge. Most classic multi-antenna receiver designs that perform multiuser detection cannot be easy decentralized because the symbol estimations depend on all information being gathered in the CPU. In this work, we consider the problem of designing \emph{fully} decentralized receivers in which the local channel state information (CSI) of RPUs is not allowed to be shared.

As the antenna arrays become extremely large (XL), a new channel effect also challenges the functionality of the receivers. The appearance of the so-called \emph{spatially non-stationary channels} consists of the fact that only part of the antenna array is containing most of the users' equipment (UE) energy due to the augmented array size \cite{Carvalho2020}. In order to allow even more scalability of spatial resolution of the XL-MIMO systems, we are also interested in designing receivers which are aware of such spatial non-stationarities.

Many recent publications address the idea of decentralizing the baseband processing in the BS for the M-MIMO systems \cite{Li2019}. A partially decentralized receiver was considered in \cite{Jeon2019}. 
However, a clear drawback of a partially decentralized method is that it requires a lot of exchange of information between the RPUs and the orchestration of the process by the CPU. Furthermore, the channel Gramian matrix has to be inverted in the CPU, with a complexity growth of the third order of the number of UEs. These disadvantages motivated the search for fully decentralized algorithms. 
The method proposed in \cite{Li2017} is fully decentralized, but suffers from the exchange of consensus information between the RPUs.
The authors of \cite{Sanchez2020} proposed an approach to fully decentralize the classic zero-forcing (ZF) design using the gradient descent (GD) method. Their scheme is based on a serial architecture that gives rise to daisy-chained RPUs. Due to the sequential property of the architecture, the \emph{GD receiver} from \cite{Sanchez2020} must perform several iterations over the RPUs to obtain a performance close to that of the centralized ZF, thus having a considerable high processing delay.

In the context of XL-MIMO systems, the authors of \cite{Amiri2021} combine the variational message-passing (MP) and the belief-propagation (BP) frameworks to propose a fully decentralized scheme called the \emph{MP-BF receiver}. Their method takes into account the spatial non-stationarities. Despite the good results, the method is more complex than the GD receiver. We believe that methods with such good performance and less complexity than the MP-BF receiver can be found for the XL-MIMO case.

\noindent \textbf{Contributions.} In this paper, we start by showing two alternative approaches to obtain the GD receiver from \cite{Sanchez2020}. These methods are based on the decentralized version of the Kaczmarz algorithm that was recently proposed in \cite{Hegde2019} and the philosophy behind successive interference cancellation (SIC) \cite{Verdu1998,Carvalho2012}. Because of the former, we will refer to the GD receiver of \cite{Sanchez2020} as the \emph{standard distributed Kaczmarz (SDK) receiver}, which are based on what is called the principle of \emph{successive residual cancellation} (SRC). We then make use of the connection with the Kaczmarz algorithm to present two ways to further improve the initial convergence of the SDK receiver. First, we propose a new receiver that decentralizes the classic ZF receiver based on the Bayesian perspective. This scheme has the advantage of estimating and eliminating noise between nodes as the iterations along the nodes are performed. 
Second, we propose a more practical method to choose the relaxation parameter of the SDK receiver than the one proposed in \cite{Sanchez2020}. Our method is based on a heuristic principle that tries to overcome the fact that the optimization of the relaxation parameter is non-trivial. Yet, it is more flexible and performs better than the previous form from \cite{Sanchez2020}. Throughout the paper, we will offer guidelines for future research that can be explored to further improve the viability of receivers based on the SRC philosophy in practice.

\noindent \textbf{Notations.}
Non-boldface letters denote scalars. Upper and lower case boldface are used for matrices and vectors. Discrete and continuous sets are given by calligraphic $\mathcal{X}$ and blackboard bold $\mathbb{X}$. The $n$-th entry of $\mathbf{x}$ is $x_n$. Matrix concatenation is $[\mathbf{X},\mathbf{Y}]$ for horizontal and $[\mathbf{X};\mathbf{Y}]$ for vertical. Transpose and Hermitian transpose are $(\cdot)^\transp$ and $(\cdot)^\htransp$. The $l_2$-norm is $\lVert\cdot\rVert^2_2$ and $\langle\cdot,\cdot\rangle$ denotes the inner product between any two vectors. The $M\times{N}$ matrix of zeros is $\mathbf{0}_{M\times N}$, whereas the identity matrix of size $N$ is $\eye{N}$. The standard basis vector $\mathbf{e}_n\in\mathbb{C}^{N\times{1}}$ is the $n$-th column of $\eye{N}$. Circularly symmetric complex-Gaussian distribution is $\mathcal{CN}(\cdot,\cdot)$. The expected value is $\mathbb{E}[\cdot]$. The product of matrices $\prod_{n=1}^{N}\mathbf{X}_n$ is set to the identity matrix if the index is out of the bound.

\section{System Model}\label{sec:systemmodel}
The system model concentrates on the uplink (UL) phase of an M-MIMO system wherein a BS equipped with $M$ antennas simultaneously serves $K$ single-antenna UEs, respectively indexed by the sets $\mathcal{M}=\{1,2,\dots,M\}$ and $\mathcal{K}=\{1,2,\dots,K\}$. The BS received signal $\mathbf{y}\in\mathbb{C}^{M\times{1}}$ is
\begin{equation}
    \mathbf{y}=\mathbf{H}\mathbf{x}+\mathbf{n}=\sum_{m=1}^{M}(\mathbf{h}^{\transp}_m\mathbf{x})\mathbf{e}_m + \mathbf{n},
    \label{eq:uplink-received-signal}
\end{equation}
where $\mathbf{H}\in\mathbb{C}^{M\times{K}}=[\mathbf{h}^{\transp}_1,\mathbf{h}^{\transp}_2,\dots, \mathbf{h}^{\transp}_M]^{\transp}$ is the channel matrix, $\mathbf{x}\in\mathbb{C}^{K\times{1}}\sim\mathcal{CN}(\mathbf{0}_{K\times{1}},p\eye{K})$ is the data signal vector with UL transmit power $p$, $\mathbf{n}\in\mathbb{C}^{M\times{1}}\sim\mathcal{CN}(\mathbf{0}_{M\times{1}},\sigma^2\eye{M})$ is the receiver noise vector with noise power $\sigma^2$.

\subsection{Channel Model}
We assume the block-fading channel model in which non-line-of-sight (NLoS) channels are considered to be frequency-flat and time-invariant \cite{Bjornson2017c}. The channel vector $\mathbf{h}_m\in\mathbb{C}^{K\times{1}}\sim\mathcal{CN}(\mathbf{0}_{K\times{1}},\mathbf{D}_m)$ encompasses the wireless links between the $K$ UEs and antenna $m\in\mathcal{M}$. The channel vectors of different antennas are considered spatially uncorrelated, \textit{i.e.}, $\mathbb{E}[\mathbf{h}_i\mathbf{h}^{\htransp}_j]=\mathbf{0}_{K\times{K}}, \ \forall i\neq{j}, \ i,j\in\mathcal{M}$. The covariance matrix $\mathbf{D}_m\in\{0,1\}^{K\times{K}}$ is a diagonal, indicator matrix that models spatial non-stationarities \cite{Ali2019,Carvalho2020}. An antenna $m$ is said to be \emph{visible} by UE $k$ if $[\mathbf{D}_m]_{k,k}=1$ and is \emph{non-visible} otherwise. We assume that the diagonal of $\mathbf{D}_m$ has at least $D_0$ nonzero entries for all $m\in\mathcal{M}$. The physical meaning of this assumption is that each antenna $m$ is contributing to the communication of at least $D_0$ UEs. Furthermore, let $D_0<D\leq{K}$ be the trace of $\mathbf{D}_m$ that is equal for all antenna $m\in\mathcal{M}$. We refer to $D$ as the \emph{number of effective UEs} served by each antenna. The spatially stationary case is $\mathbf{D}_m=\eye{K}, \ \forall m\in\mathcal{M}$, which characterizes the M-MIMO regime of compact antenna arrays. The matrix $\mathbf{D}_m\neq\eye{K}$ then introduces sparsity into $\mathbf{h}_m$ in the spatially non-stationary case. We are interested to evaluate how this sparsity affects the performance of the decentralized receivers discussed here. 

\section{Background}\label{sec:background}
In this section, we first present the basic structure of the distributed baseband processing architecture with daisy-chained RPUs and the basic principles behind the design of decentralized linear receivers. Next, we review the GD receiver from \cite{Sanchez2020} presenting two different ways to obtain this receiver that improve the understanding of the method.

\subsection{Daisy-Chain Baseband Processing Architecture}
We adopt the fully decentralized architecture for baseband processing at the BS considered in \cite{Sanchez2020}. For simplicity, the hardware architecture comprises $M$ RPUs, each of which manages the processing tasks concerning an antenna $m\in\mathcal{M}$.\footnote{A more practical assumption would be to group the signal processing of some antennas in an RPU, since there are some common tasks on different antennas.} The $m$-th RPU is serially connected to the $(m+1)$-th RPU, which gives rise to a daisy-chain connection layout. Only RPU $1$ has a dedicated physical link to the CPU. Figure \ref{fig:daisy-chain} gives a graphical representation of the architecture. The undirected graph consists of $M$ nodes representing the $M$ RPUs and antennas. The flat-tree is rooted at node $1$. The local knowledge available at the $m$-th node includes perfect CSI and the signal received by the $m$-th antenna, which are embraced by the tuple $(\mathbf{h}_m,y_m)$ in Fig. \ref{fig:daisy-chain}. Antennas and RPUs give physical meaning to the signals and hardware components involved in the signal reception task, while graph notation allows us to describe the distributed architecture in Fig. \ref{fig:daisy-chain} mathematically. Henceforth we will use the words antennas and nodes interchangeably, as a consequence of our simplifying assumption of associating an RPU for each antenna.

\subsection{Distributed Linear Receivers}
The design of the receiver comprehends two phases of signal processing that result respectively in a soft-estimate and a hard-estimate of $\mathbf{x}$. The soft-estimate $\hat{\mathbf{x}}\in\mathbb{C}^{K\times 1}$ can be obtained as \cite{Carvalho2012,Bjornson2017c}: 
\begin{equation}
    \hat{\mathbf{x}}=\mathbf{V}^\htransp\mathbf{y}=\sum_{m=1}^{M}\mathbf{v}^{*}_m y_m=\sum_{m=1}^{M}\hat{\mathbf{x}}_m,
    \label{eq:linear-receiver}
\end{equation}
where $\mathbf{V}\in\mathbb{C}^{M\times{K}}=[\mathbf{v}^{\transp}_1,\mathbf{v}^{\transp}_2,\dots,\mathbf{v}^{\transp}_M]^{\transp}$ is the linear combining matrix. 
The objective of the $m$-th RPU is to obtain a reliable local soft-estimate $\hat{\mathbf{x}}_m$ based on the observation $y_m$ and knowing the local CSI $\mathbf{h}_m$ and given the two main constraints imposed by the decentralized architecture in Fig. \ref{fig:daisy-chain}: (C.1) \textit{Scalability.} Information exchange between RPUs must scale with $K$ instead of $M$. This discourages the exchange of local CSI $\mathbf{h}_m$ between RPUs. (C.2) \textit{Autonomy.} Only the $m$-th RPU should be responsible for extracting relevant information from the local CSI $\mathbf{h}_m$. After the RPUs converge to a final soft-estimate $\hat{\mathbf{x}}$, the CPU demodulates $\hat{\mathbf{x}}$ to produce the hard-estimate. Given the above constraints, we will focus on obtaining a distributed method to obtain the final soft-estimate $\hat{\mathbf{x}}$. 

\begin{figure}[htbp!]
    \centering
    \includegraphics[width=\columnwidth]{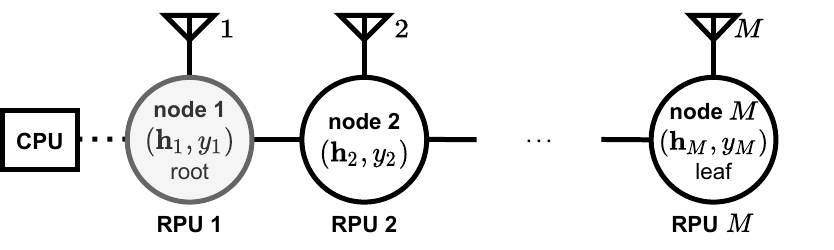}
    \caption{Decentralized baseband processing architecture at the BS following a daisy-chain arrangement. }
    \label{fig:daisy-chain}
\end{figure}

\subsection{Standard Distributed Kaczmarz (SDK) Receiver}\label{subsec:standard}
The ZF combiner can be obtained by solving the following unconstrained least-squares (LS) problem \cite{Verdu1998}:
\begin{equation}
    \argmin_{\mathbf{x}\in\mathbb{C}^{K\times{1}}}\lVert\mathbf{y}-\mathbf{H}\mathbf{x}\rVert^2_2.
    \label{eq:unconstrained-least-squares}
\end{equation}
We look at \eqref{eq:unconstrained-least-squares} as the the overdetermined system of linear equations (SLE) $\mathbf{y}=\mathbf{H}\mathbf{x}$, which is inconsistent due to noise disturbance in $\mathbf{y}$. Furthermore, $\mathbf{H}$ is full-rank in the spatially stationary case and can be rank-deficient\footnote{The sparsity in $\mathbf{H}$ can introduce the possibility that at least two channel vectors yields a linearly dependent set \cite{Ali2019}.} in the spatially non-stationary case. In our setting, due to (C.1) and (C.2), 
the $m$-th node only knows the $m$-th equation $y_m=\mathbf{h}^{\transp}_m\mathbf{x}$. A popular iterative algorithm that can solve $\mathbf{y}=\mathbf{H}\mathbf{x}$ is the Kaczmarz algorithm \cite{Kaczmarz1937}. Indeed, the modified Kaczmarz method introduced in \cite{Hegde2019} considers exactly the case of interest where the equations are indexed by the nodes of a tree. The method is divided into the \emph{dispersion} and \emph{pooling} stages, as discussed in the sequel, and referred to as the SDK receiver.

\vspace{2mm}
\noindent\textbf{Dispersion stage.}
The root node starts with an initial guess $\hat{\mathbf{x}}^{(t)}_0$, where $t\in\{1,2,\dots,T\}$ indexes \emph{cycles} over the tree. This initial guess is improved by iterative steps that are defined by node $m\in\mathcal{M}$ receiving from its immediate predecessor $m-1$ an input estimate $\hat{\mathbf{x}}^{(t)}_{m-1}$. Node $m$ then applies the Kaczmarz update \cite{Kaczmarz1937} to generate the following new estimate:%
\begin{IEEEeqnarray}{rCl}
    {r}^{(t)}_m&=&y_m-\mathbf{h}^{\transp}_m\hat{\mathbf{x}}^{(t)}_{m-1},\IEEEnonumber\\
    \hat{\mathbf{x}}^{(t)}_m&=&\hat{\mathbf{x}}^{(t)}_{m-1}+\lambda_m\mathbf{h}^{*}_m{r}^{(t)}_m,
    \label{eq:kaczmarz-iterative-step}
\end{IEEEeqnarray}
where $r^{(t)}_m$ is the \emph{residual} and $\lambda_m=\lambda\lVert{\mathbf{h}_m}\rVert^{-2}_2$ with $\lambda$ being the \emph{relaxation parameter}. The relaxation parameter can help in handling the inconsistency of the SLE \cite{Censor1983} and is discussed in depth in Subsection \ref{subsec:relaxation}. Fig. \ref{fig:sdk-receiver} shows a block diagram illustrating the above iterative step. The process continue recursively with node $m$ transferring $\hat{\mathbf{x}}^{(t)}_m$ to its immediate successor $m+1$. The dispersion stage finishes when the leaf node obtains its estimate $\hat{\mathbf{x}}^{(t)}_M$.

\begin{figure}[htbp!]
    \centering
    \includegraphics[width=\columnwidth]{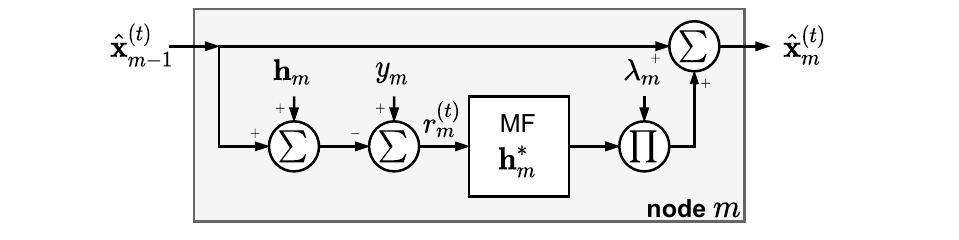}
    \caption{Kaczmarz iterative step in eq. \eqref{eq:kaczmarz-iterative-step} of the SDK receiver at node $m\in\mathcal{M}$. MF stands for matched filter.} 
    \label{fig:sdk-receiver}
\end{figure}

\vspace{2mm}
\noindent\textbf{Pooling stage.} Since we are considering a flat-tree in Fig. \ref{fig:daisy-chain}, the pooling stage simply proceeds with the leaf node backpropagating its estimate $\hat{\mathbf{x}}^{(t)}_M$ through the tree without updating and until it reaches the root node. A new cycle $t+1$ over the tree can begin by defining $\hat{\mathbf{x}}^{(t+1)}_0=\hat{\mathbf{x}}^{(t)}_M$; otherwise, $\hat{\mathbf{x}}=\hat{\mathbf{x}}^{(t)}_M$ is transferred to the CPU for demodulation. In \cite{Sanchez2020}, a dedicated physical connection between the root and leaf nodes was considered to reduce communication delay. However, we chose to consider here the most general case of backpropagation because this naturally allows the generalization of the algorithm for more intricate tree's topologies, as we will discuss in Subsection \ref{sucsec:intricate-trees}. Fig. \ref{fig:dispersion-pooling-stages} offers an overview of the dispersion and pooling stages. 
\begin{figure}[htbp!]
    \centering
    \includegraphics[width=\columnwidth]{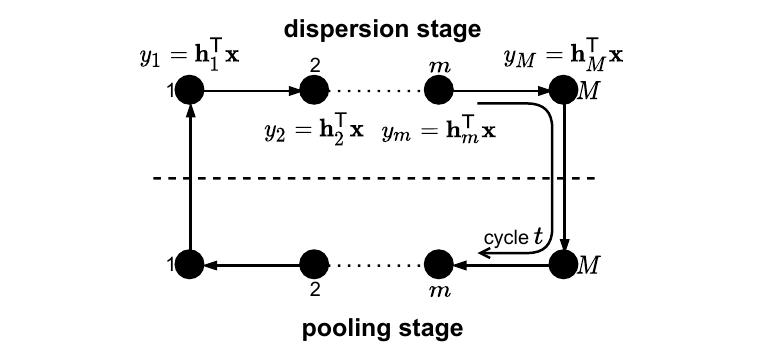}
    \vspace{-8mm}
    \caption{Illustration of the dispersion and pooling stages. The pooling stage is depicted by mirroring the nodes. The algebraic equations of $\mathbf{y}=\mathbf{H}\mathbf{x}$ are indexed by the nodes of the flat-tree in Fig. \ref{fig:daisy-chain}. A cycle $t\in\{1,2,\dots,T\}$ can be seen as the loop defined by the sequential realization of dispersion and pooling stages.}
    \label{fig:dispersion-pooling-stages}
\end{figure}

\vspace{2mm}
\noindent\textit{Remark 1 (Convergence of the SDK receiver).} 
For an inconsistent SLE, the SDK receiver converges to a weighted LS solution that depends on $\lambda$ in the case of a full-rank $\mathbf{H}$ (stationary case) and to a weighted minimum norm LS solution in the case of a rank-deficient $\mathbf{H}$ (non-stationary case) \cite{Hegde2019}.

\vspace{2mm}
\noindent\textit{Remark 2 (Complexity of the SDK receiver).} Each node $m\in\mathcal{M}$ only needs to store a $K$-dimensional complex vector $\hat{\mathbf{x}}^{(t)}_m$. For the $m$-th node, the Kaczmarz iterative step in \eqref{eq:kaczmarz-iterative-step} requires $(12K+2)T$ floating-point per second (FLOPs) measured in terms of real operations, where $T$ is the total number of cycles. The overall information exchange between any two nodes scales with $4KT$ FLOPs due to the pooling stage with backpropagation.

We note that the iterative step of the SDK receiver in eq. \eqref{eq:kaczmarz-iterative-step} is essentially equivalent to the one introduced in eq. (11) of \cite{Sanchez2020}, where the authors approximated the GD method to cover constraints (C.1) and (C.2). 
The connection to the Kaczmarz algorithm is beneficial because of the convergence guarantees in \cite{Hegde2019}. Moreover, we will borrow ideas from the vast body of literature that has studied ways to improve the convergence of Kaczmarz-based approaches \cite{Hegde2019,Herman2009,Strohmer2009,Censor1983} in Section \ref{sec:contributions}.

The authors of \cite{Sanchez2020} exploited the intrinsic recursion of the Kaczmarz update step in \eqref{eq:kaczmarz-iterative-step} to calculate the combining vector $\mathbf{v}_m$. For $T=1$, the leaf estimate $\hat{\mathbf{x}}^{(1)}_M$ can be written as:
\begin{IEEEeqnarray}{rCl}
    \hat{\mathbf{x}}^{(1)}_M&=&\prod_{m=1}^{M}(\eye{K}-\lambda_m\mathbf{h}_m\mathbf{h}^{\htransp}_m)\hat{\mathbf{x}}^{(1)}_0\IEEEnonumber\\
    &+&\sum_{m=1}^{M}\prod_{i=m+1}^{M}(\eye{K}-\lambda_i\mathbf{h}_i\mathbf{h}^{\htransp}_i)\lambda_m\mathbf{h}^{*}_m y_m,
    \label{eq:xM-SDK}
\end{IEEEeqnarray}
which encounters structural similarity with eq. \eqref{eq:linear-receiver}; assuming that the initial guess is $\hat{\mathbf{x}}^{(1)}_0=\mathbf{0}_{K\times{1}}$, $\mathbf{v}_m=\prod_{i=m+1}^{M}(\eye{K}-\lambda_i\mathbf{h}_i\mathbf{h}^{\htransp}_i)\lambda_m\mathbf{h}_m$. 
Expressions like the above will be used to draw some conclusions throughout the paper and make it explicit the straightforwardly generalization of the receivers to obtain $\mathbf{v}_m$ in a distributed manner.

\subsection{Successive Residual Cancellation (SRC) Receiver}
An alternative approach in obtaining the soft-estimate $\hat{\mathbf{x}}$ is to apply the heuristic idea of the linear SIC receiver. The principle of SIC can be seen as a two step operation \cite{Verdu1998,Carvalho2012}: (S.1) one soft-estimate is estimated at a time and (S.2) its contribution is removed from the received signal. We want to apply these steps to combine the local soft-estimates $\hat{\mathbf{x}}_m$ between the different nodes in the distributed architecture of Fig. \ref{fig:daisy-chain}. To this end, an observation is in order. In the conventional multiuser detection problem \cite{Verdu1998,Carvalho2012}, the SIC suppresses inter-user interference. In the new distributed processing problem, each node $m\in\mathcal{M}$ only knows $y_m$ and $\mathbf{h}_m$ and gets from these a local soft-estimate $\hat{\mathbf{x}}_m$. The nodes have then to agree on a final soft-estimate $\hat{\mathbf{x}}$ of $\mathbf{x}$. Therefore, the interference notion of the multiuser setting is translated into the residual difference between $\hat{\mathbf{x}}_m$'s of the nodes.


With this observation, we apply the SIC steps to successively eliminate the residuals, giving rise to the following \emph{successive residual cancellation} (SRC) receiver:\footnote{The SRC principle goes back to the idea of solving the SLE by operating on one equation at a time \cite{Kaczmarz1937}.} 
\begin{IEEEeqnarray}{rCl}
    \hat{\mathbf{x}}_m =\mathbf{w}^{*}_m y_m + \sum_{i=1}^{m-1}\hat{\mathbf{x}}_i&,& \qquad  \text{from (S.1)},\IEEEnonumber\\
    \mathbf{y}^{(m)}=\mathbf{y}^{(m-1)}-(\mathbf{w}^{\transp}_m\hat{\mathbf{x}}_{m-1})\mathbf{e}_m, \mathbf{y}^{(1)}=\mathbf{y}&,&\qquad \text{from (S.2)}, \quad  \label{eq:sic-receiver}
\end{IEEEeqnarray}
where $\mathbf{w}_m\in\mathbb{C}^{K\times 1}$ stands for a linear combiner and $\mathbf{y}^{(m)}$ denotes the residual, that is, the received signal $\mathbf{y}$ less the estimates $\hat{\mathbf{x}}_i$ from the previous $m-1$ iterates. From a straightforward matrix-algebra analysis over the recursion of $\mathbf{y}^{(m)}$, we have the following soft-decision estimate after $M$ SRC iterations:
\begin{equation}
    \hat{\mathbf{x}}_M=\sum_{m=1}^{M}\prod_{i=m+1}^{M}(\eye{K}-\mathbf{w}_i\mathbf{w}^{\htransp}_i)\mathbf{w}^{*}_m y_m,
    \label{eq:xM-SIC-antenna}
\end{equation}
which is mathematically equivalent to the expression obtained for the SDK receiver in \eqref{eq:xM-SDK} when the normalized matched-filter (MF) is chosen $\mathbf{w}_m=\lVert{\mathbf{h}_m}\rVert^{-2}_2\mathbf{h}_m$, zero vector is the initial guess, and $\lambda=1$. 
In fact, the SDK receiver can be seen as a generalization of the SRC due to the relaxation parameter $\lambda$. In \cite{Bentrcia2019}, a similar connection was found between the Kaczmarz algorithm and the SIC receiver in the usual multiuser setting.




\section{Accelerating Initial Convergence}\label{sec:contributions}
In this section, motivated by the connections made with the Kaczmarz algorithm and the SRC principle, we discuss different ways to improve the initial convergence of the SDK receiver. The interest in improving its initial behavior lies in the fact that we naturally want the algorithm to end as soon as possible, thus reducing the complexity and latency to obtain a final soft-estimate $\hat{\mathbf{x}}$.


\subsection{Bayesian Distributed Kaczmarz (BDK) Receiver}\label{subsec:bayesian}
One way to handle the inconsistency of $\mathbf{y}=\mathbf{H}\mathbf{x}$ is to adopt the Bayesian perspective for finding ${\mathbf{x}}$. Given a realization of the channel matrix $\mathbf{H}$, the idea is to incorporate prior\footnote{In the event that other practical information is known about the $\mathbf{x}$ and $\mathbf{n}$ vectors, processes known as tricks can be inserted to further improve convergence performance. For example, in the case that $\mathbf{x}$ is drawn from a 16-QAM digital modulation scheme. See \cite{Herman2009} for an overview of these tricks.} knowledge of the data signal $\mathbf{x}$ and receiver noise $\mathbf{n}$ from \eqref{eq:uplink-received-signal} to the estimation problem. We make the following assumptions: (A.1) either $\mathbf{x}$ and $\mathbf{n}$ follows complex circularly-symmetric central Gaussian distributions, which is justified by a worst-case scenario from the point of view of communication \cite{Bjornson2017c}; (A.2) the entries of either $\mathbf{x}$ and $\mathbf{n}$ are uncorrelated and have the same variance (UL transmit power ${p}$ and noise power $\sigma^2$, respectively), which is justified respectively by: each UE generates its data independently; each antenna sensor experiences noise independently. This gives rise to the following $l_2$-constrained LS problem that maximizes the \textit{a posteriori} probability \cite{Herman2009}:
\begin{equation}
    \argmin_{\mathbf{x}\in\mathbb{C}^{K\times{1}}}\lVert\mathbf{y}-\mathbf{H}\mathbf{x}\rVert^2_2+\xi\lVert{\mathbf{x}}\rVert^2_2,
    \label{eq:constrained-least-squares}
\end{equation}
where $\xi=\sigma^2/p$ is the inverse of the pre-processing signal-to-noise (SNR) ratio. Some readers familiar with the regularization method may recognize that the problem above can also be interpreted from the point of view of the Tikhonov regularization \cite{Shai2014} and its solution is related to the regularized zero-forcing (RZF) combiner \cite{Bjornson2017c,Verdu1998}. What we essentially have now from \eqref{eq:constrained-least-squares} is the following underdetermined SLE $\mathbf{y}=\mathbf{H}\mathbf{x}+\sqrt{\xi}\mathbf{u}=\mathbf{B}\mathbf{z}$, where $\mathbf{B}\in\mathbb{C}^{M\times{(K+M)}}=[\mathbf{H},\sqrt{\xi}\eye{M}]$ and $\mathbf{z}\in\mathbb{C}^{(K+M)\times{1}}=[{\mathbf{x}\in\mathbb{C}^{K\times{1}};\mathbf{u}}\in\mathbb{C}^{M\times1}]$ \cite{Herman2009}. Different from the previous SLE, $\mathbf{y}=\mathbf{B}\mathbf{z}$ is consistent due the augmented range brought by $\mathbf{u}$. Moreover, $\mathbf{B}$ is full-rank in the spatially stationary case and can be rank-deficient in the spatially non-stationary case. Below, we describe the distributed Kaczmarz algorithm \cite{Hegde2019} applied to solve $\mathbf{y}=\mathbf{B}\mathbf{z}$ and omit repetitive definitions. We refer to the procedure obtained as the \emph{Bayesian distributed Kaczmarz} (BDK) receiver.

\vspace{2mm}
\noindent\textbf{Dispersion stage.} Initial guesses $\hat{\mathbf{x}}^{(t)}_{0}$ and $\hat{\mathbf{u}}^{(t)}_{0}$ are available at the root node. The Kaczmarz iterative step at node $m$ is now:
\begin{IEEEeqnarray}{rCl}
{r}^{(t)}_m&=&y_m-\mathbf{h}^{\transp}_m\hat{\mathbf{x}}^{(t)}_{m-1}-\sqrt{\xi}\mathbf{e}^{\transp}_m\hat{\mathbf{u}}^{(t)}_{m-1}\IEEEnonumber,\\
\hat{\mathbf{x}}^{(t)}_{m}&=&\hat{\mathbf{x}}^{(t)}_{m-1}+\lambda^{\star}_m\mathbf{h}^{*}_m{r}^{(t)}_m,\IEEEnonumber\\
\hat{\mathbf{u}}^{(t)}_{m}&=&\hat{\mathbf{u}}^{(t)}_{m-1}+\lambda^{\star}_m\sqrt{\xi}\mathbf{e}^{*}_m{r}^{(t)}_m,
\label{eq:bayesian-kaczmarz-iterative-step}
\end{IEEEeqnarray}
where $\lambda^{\star}_m=\lambda^{\star}(\lVert\mathbf{h}_m\rVert^2_2+\xi)^{-1}$ is the relaxation parameter regarding node $m\in\mathcal{M}$. Note that $\lambda^{\star}$ of the BDK receiver is different from $\lambda$ of the SDK. Since the problem being solved is now consistent, we will set $\lambda^{\star}=1$.\footnote{Future work can discuss a more suitable choice.} 
Fig. \ref{fig:bdk-receiver} illustrates the new Kaczmarz iterative step. The dispersion stage ends when the leaf node gets $\hat{\mathbf{x}}^{(t)}_{M}$ and $\hat{\mathbf{u}}^{(t)}_{M}$.

\begin{figure}[htbp!]
    \centering
    \includegraphics[width=\columnwidth]{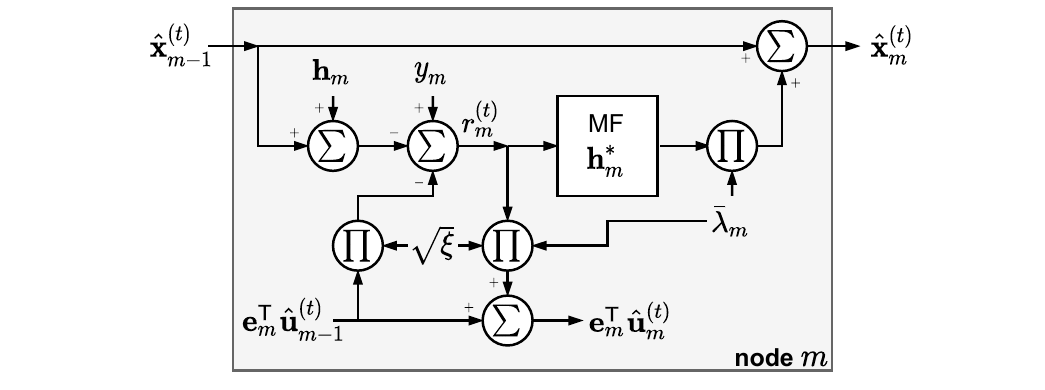}
    \caption{Kaczmarz iterative step in eq. \eqref{eq:bayesian-kaczmarz-iterative-step} of the BDK receiver at node $m\in\mathcal{M}$.}
    \label{fig:bdk-receiver}
\end{figure}

\vspace{2mm}
\noindent\textbf{Pooling stage.} At first sight, the BDK receiver with iterative step defined by \eqref{eq:bayesian-kaczmarz-iterative-step} seems to make little sense. The algorithm description suggests that $\hat{\mathbf{u}}^{(t)}_{m}$ should be exchanged between nodes, whose vector is in $\mathbb{C}^{M\times{1}}$; therefore, violating (C.1). However, a careful look at the iterative step shows that the $m$-th component of $\hat{\mathbf{u}}^{(t)}_{m}$ is only altered by the $m$-th node because of $\mathbf{e}_m$. As a consequence, $\hat{\mathbf{u}}^{(t)}_{M}$ does not need to be backpropagated to the root node. 
To gain further intuition about this, note that $\sqrt{\xi}\mathbf{e}^{\transp}_{m}\hat{{\mathbf{u}}}_{m-1}$ is trivially an estimate of $n_m$, which is the receiver noise relative to the $m$-th antenna in \eqref{eq:uplink-received-signal}. Hence, here the pooling stage progresses akin to that of the SDK algorithm, as shown is Fig. \ref{fig:dispersion-pooling-stages}.

\vspace{2mm}
\noindent\textit{Remark 3 (Convergence of the BDK receiver).} For a consistent SLE, the BDK receiver converges to the minimum norm LS solution in either stationary and non-stationary cases \cite{Hegde2019}.

\vspace{2mm}
\noindent\textit{Remark 4 (Complexity of the BDK receiver).} Each node $m\in\mathcal{M}$ requires fixed storage for a $K$-dimensional complex vector $\hat{\mathbf{x}}^{(t)}_{m}$ and a complex scalar representing the $m$-th entry of $\hat{\mathbf{u}}^{(t)}_{m}$. For the $m$-th node, the Kaczmarz iterative step in \eqref{eq:bayesian-kaczmarz-iterative-step} requires $(12K+6)T$ FLOPs. The overall information exchange between any two nodes scales with $4KT$ FLOPs.

From Remarks 2 and 4, it is easy to see that the computational requirements of the SDK and BDK receivers are essentially the same thanks to the benign property involving the distributed knowledge of $\hat{\mathbf{u}}^{(t)}_{M}$ in the nodes. The remarks also comprise an upper bound of complexity, since they consider $\mathbf{h}_m$ dense. Further computational gains are obtained when considering $\mathbf{h}_m$ sparse \cite{Herman2009}. 

\subsubsection{Understanding BDK receiver} There are two key differences between the SDK and BDK receivers. The first is the incorporation of prior knowledge $\xi$ in $\lambda^{\star}_m=\lambda^{\star}(\lVert\mathbf{h}_m\rVert^2_2+\xi)^{-1}, \ \forall m \in\mathcal{M}$, which helps to combat the noise propagation through the residuals. The second is a change on the way that the estimates are combined through the nodes due to the addition of the noise estimation. For intuition about this, we use part of the recursion of the Kaczmarz iterative step in \eqref{eq:bayesian-kaczmarz-iterative-step} to generate the following expressions for $\hat{\mathbf{x}}^{(1)}_{M}$ and $ \hat{\mathbf{u}}^{(1)}_{M}$:
\begin{IEEEeqnarray}{rCl}
    \hat{\mathbf{x}}^{(1)}_M&=&\prod_{m=1}^{M}(\eye{K}-\lambda^{\star}_m\mathbf{h}_m\mathbf{h}^{\htransp}_m)\hat{\mathbf{x}}^{(1)}_0\IEEEnonumber\\
    &+&\sum_{m=1}^{M}\prod_{i=m+1}^{M}(\eye{K}-\lambda^{\star}_i\mathbf{h}_i\mathbf{h}^{\htransp}_i)\lambda^{\star}_m\mathbf{h}^{*}_m( y_m-\sqrt{\xi}\mathbf{e}^{\transp}_m\hat{\mathbf{u}}^{(t)}_{m-1}),\IEEEnonumber\\
    \hat{\mathbf{u}}^{(1)}_M&=&\prod_{m=1}^{M}(\eye{M}-\lambda^{\star}_m\xi\mathbf{e}_m\mathbf{e}^{\htransp}_m)\hat{\mathbf{u}}^{(1)}_0\IEEEnonumber\\
    &+&\sum_{m=1}^{M}\prod_{i=m+1}^{M}(\eye{M}-\lambda^{\star}_i\xi\mathbf{e}_i\mathbf{e}^{\htransp}_i)\lambda^{\star}_m\sqrt{\xi}\mathbf{e}^{*}_m( y_m-\mathbf{h}^{\transp}_m\hat{\mathbf{x}}^{(t)}_{m-1}).\IEEEnonumber
\end{IEEEeqnarray}
Comparing these with eq. \eqref{eq:xM-SDK}, it is possible to note that the estimates now rely on "effective" received signals. For example, to estimate a new $\hat{\mathbf{x}}^{(t)}_m$, the effective signals of the form $y_i-\sqrt{\xi}\mathbf{e}^{\transp}_i\hat{\mathbf{u}}^{(t)}_{i-1}$ for $i\in\mathcal{M}, i\leq m$ will be considered, which eliminates the current noise estimate $\sqrt{\xi}\mathbf{e}^{\transp}_i\hat{\mathbf{u}}^{(t)}_{i-1}$ from the true received received signal $y_i$.

\subsection{Choosing the Relaxation Parameter for the SDK Receiver}\label{subsec:relaxation}
The relaxation parameter can significantly improve the performance of the Kaczmarz algorithm in practice. However, it is very difficult to find a suitable relaxation parameter, since its optimization is non-trivial \cite{Herman2009}. 
The authors in eq. (22) of \cite{Sanchez2020} proposed the choice $\lambda^{\dagger}=\frac{1}{2}\frac{K}{M}\log(4\cdot M \cdot \mathrm{SNR})$ for the SDK receiver. Despite the good results, this value comes from the approximation of a closed-form expression of the signal-to-interference-plus-noise ratio (SINR) for $\hat{\mathbf{x}}^{(1)}_M$ expressed in eq. \eqref{eq:xM-SDK}. Therefore, this approach is very inflexible in the sense that we need to approximate the SINR for different cycles $t$ and new approximations are necessary for different channel models. In this part, we propose a more practical and effective way to define the relaxation parameter resulting in better, although sub-optimal yet, performance than the approach proposed in \cite{Sanchez2020}. Our method is a heuristic based on the combination of the following three observations:\footnote{There are some approaches that choose the relaxation parameters based on the normalization of the eigenvalues of $\mathbf{H}^{\htransp}\mathbf{H}$ and improvement of its condition number \cite{Hanke1990}. However, due to the decentralized constraints (C.1) and (C.2), we do not have access to $\mathbf{H}^{\htransp}\mathbf{H}$ let alone its inverse.} 

\noindent (O.1) \textit{Connection to GD.} In \cite{Shai2014}, a simple and reasonable mathematical framework to select the step-size of the GD method was given. The step-size of the GD mathematically resembles $\lambda$ of the Kaczmarz iterative step in \eqref{eq:kaczmarz-iterative-step}. 

\noindent (O.2) \textit{Noise Aggregation.} In \cite{Elfving2014}, it was shown that the propagation of the noise error embedded in the residual is proportional to $\sqrt{m}$, where $m\in\mathcal{M}$ represents the current Kaczmarz iterative step applied by node $m$. This intuitively makes sense in the \emph{case of dense} $\mathbf{h}_m$, since $\hat{\mathbf{x}}^{(t)}_{m-1}$ suffers from the aggregation of $m-1$ different and independent noise sources of the type $n_{i}\sim\mathcal{CN}(0,\sigma^2), \ i<m, \ \forall i\in\mathcal{M}$. This suggest the definition of a $\lambda$ that is a function of $m$ and decreases according to $\sqrt{m}$. Since there is less information in future iterations, we can expand this idea to also provide a $\lambda$ that decreases with the number of cycles in the form $\sqrt{t}$.

\noindent (O.3) \textit{Underrelaxation.} It has been found in \cite{Herman2009,Censor1983} that the use of $0<\lambda<1$ improves the performance of the Kaczmarz-based methods when solving inconsistent SLEs. 

Our proposal is to chose $\lambda$ based on the mathematical framework provided by (O.1) and incorporate (O.2) and (O.3) to the solution. First, we obtain the following result.

\begin{theorem}\label{theorem:01}
Let $\mathbf{x}$ be the solution of \eqref{eq:unconstrained-least-squares} and $\boldsymbol{\varrho}^{(t)}_m=\mathbf{h}^{*}_m{r}^{(t)}_m$. The Kaczmarz iterative step in \eqref{eq:kaczmarz-iterative-step} can then be rewritten as
\begin{equation}
    \hat{\mathbf{x}}^{(t)}_m=\hat{\mathbf{x}}^{(t)}_{m-1} + \lambda_m \boldsymbol{\varrho}^{(t)}_m.
    \label{eq:theorem:sse:01}
\end{equation}
If $t=1$ and $\hat{\mathbf{x}}^{(1)}_{0}=\mathbf{0}_{K\times{1}}$, the expected sum of the squared errors of the form $\lVert\hat{\mathbf{x}}^{(1)}_{m}-\mathbf{x}\rVert^{2}_{2}$ conditioned on the set $\{\mathbf{h}_m\}, \ \forall m \in\mathcal{M}$ that represents perfect CSI knowledge can be bounded by the following quantity:
\begin{IEEEeqnarray}{rCl}
    &&\mathbb{E}\left[\sum_{m=1}^{M}\mathcal{R}(\langle\hat{\mathbf{x}}^{(1)}_{m-1}-\mathbf{x}, \boldsymbol{\varrho}^{(1)}_m\rangle)|\{\mathbf{h}_m\}\right]\geq-\dfrac{1}{2\lambda_1}\mathbb{E}[\ltwonorm{\mathbf{x}}|\{\mathbf{h}_m\}]\IEEEnonumber\\
    &&-\sum_{m=1}^{M}\dfrac{\lambda_m}{2}\mathbb{E}[\ltwonorm{\boldsymbol{\varrho}^{(1)}_m}|\{\mathbf{h}_m\}],
    \label{eq:theorem:sse:02}
\end{IEEEeqnarray}
where expectation is taken with respect to $\mathbf{x}$ and $\mathbf{n}$. In particular, if we specify that $\mathbb{E}[\ltwonorm{\mathbf{x}}|\{\mathbf{h}_m\}]\geq B^2$, $\mathbb{E}[\ltwonorm{\boldsymbol{\varrho}^{(1)}_m}|\{\mathbf{h}_m\}]\geq\rho^2$, for every $B,\rho>0$ for all $m\in\mathcal{M}$, and if we define $\lambda_m=\sqrt{\frac{B^2}{\rho^2 m}}\lVert\mathbf{h}_m\rVert^{-2}_{2}$, we have that
\begin{IEEEeqnarray}{rCl}
    \mathbb{E}\left[\sum_{m=1}^{M}\mathcal{R}(\langle\hat{\mathbf{x}}^{(1)}_{m-1}-\mathbf{x}, \boldsymbol{\varrho}^{(1)}_m\rangle)|\{\mathbf{h}_m\}\right]&\geq&-\dfrac{\ltwonorm{\mathbf{h}_1}B\rho}{2}\IEEEnonumber\\
    &-&\dfrac{B\rho}{2}\sum_{m=1}^{M}\dfrac{\lVert{\mathbf{h}_m}\rVert^{-2}_2}{\sqrt{m}}.
    \label{eq:theorem:sse:03}
\end{IEEEeqnarray}
Given the distributions of $\mathbf{x}$ and $\mathbf{n}$ in \eqref{eq:uplink-received-signal}, loose choices are $B^2=Kp$ and $\rho^2=\sigma^2$ for the SDK receiver and dense $\mathbf{h}_m, \ \forall m\in\mathcal{M}$.
\end{theorem}
\begin{proof}
    The proof is given in Appendix A.
\end{proof}

Theorem $\ref{theorem:01}$ relates the choice of the relaxation parameter to a quantity of interest that measures how the expectation of the sum of the squared errors of the type $\lVert\hat{\mathbf{x}}^{(1)}_{m}-\mathbf{x}\rVert^{2}_{2}$ progresses through the nodes. In \eqref{eq:theorem:sse:03}, we already have considered part of (O.2) by incorporating the factor $\sqrt{m}$ into the parameter choice. Plugging (O.2), (O.3), and the result of the application of (O.1) in Theorem \ref{theorem:01} in \eqref{eq:theorem:sse:03}, we can heuristically set
\begin{equation}
    \lambda_m=\underbrace{\min\left[\sqrt{\frac{K\cdot\mathrm{SNR}}{t\cdot m}},\, 1\right]}_{=\lambda_{m,t}}\lVert\mathbf{h}_m\rVert^{-2}_{2},
    \label{eq:lambda:sdk:proposed}
\end{equation}
where $\mathrm{SNR}=p/\sigma^2$. 
The proposed relaxation parameter {${\lambda}_{m,t}$} is a function of, according to (O.2), $m\in\mathcal{M}$ and $t\in\{1,2,\dots,T\}$. The minimum function ensures that underrelaxation is applied in each node $m$ due to (O.3). 



\subsection{Ordering and Randomization}\label{subsec:ord-rand}
In the context of SIC receivers, it is well known that the order the UEs are detected can affect the overall reliability of the successive detector due to error propagation \cite{Verdu1998,Carvalho2012}. A good strategy is then to detect the UEs in the order of decreasing post-processing SNRs. In the decentralized setting, this observation also applies for the SRC-based receivers: the order in which local estimates $\hat{\mathbf{x}}_m$ are obtained at each node $m\in\mathcal{M}$ can increase the risk of error propagation. In the context of the Kaczmarz algorithm, this problem is also known for the fact that an efficient order of the equations of the SLE being solved can significantly improve the initial behavior of the algorithm in practice \cite{Herman2009}. Since it is often difficult to obtain the optimum order, a suboptimal approach consists in randomizing the choice of the equations based on their energies \cite{Strohmer2009}. In our setting, however, the nodes are physically connected sequentially, which prevents us from ordering how the local estimates are obtained without adding a moderate level of coordination by the CPU, while aggregating an overhead to the process. Motivated by \cite{Strohmer2009}, we partially randomized the order by selecting the root node of the flat-tree in Fig. \ref{fig:daisy-chain} based on probabilities proportional to $\ltwonorm{\mathbf{h}_m}$. Although small improvements could be seen in high SNR, the results obtained were considered not very promising. Further study is needed to comprehend how the wired restriction affects the error propagation in SRC-based receivers. In the XL-MIMO case, this discussion is even more important, since the received energy is unequal across the nodes. 

\subsection{Tree-Based Baseband Processing Architectures}\label{sucsec:intricate-trees}
To illustrate the generalization of the SDK receiver to tree-based processing architectures, we rely on the concept of sub-arrays introduced in \cite{Carvalho2020}. A sub-array embraces part of the $M$ antenna elements that forms the whole antenna array deployed at the BS in which the channel can be considered spatially stationary. Fig. \ref{fig:coarse-tree} illustrates a decentralized processing architecture with a sub-array $s_i$ being represented by a \emph{node common bus}, where $i\in\{1,2,\dots,S\}$ and $S$ is the total number of subarrays. The sub-arrays are connected to the \emph{sub-array common bus}, which can be seen as the root $r$ of a tree. The leaf nodes $l_m, \ m\in\mathcal{M}$ represents the $M$ RPUs and have as information one equation $y_m=\mathbf{h}_m^{\transp}\mathbf{x}$ of the SLE $\mathbf{y}=\mathbf{H}\mathbf{x}$. There are now weights of the form $g(n,v)=g(v,n)$, with $g(n,v)>0$, associated to the links between $n$ and $v$, which denote different elements of the architecture.


Two changes are needed to define the SDK algorithm described previously for the structure in Fig. \ref{fig:coarse-tree}. First, in the dispersion stage, the initial guess $\hat{\mathbf{x}}^{(t)}_r$ is also spread over the structure with the detail that the same information is accessible in the buses. Second, in the pooling stage, the leaf nodes $l_m$'s also need to backpropagate their $\hat{\mathbf{x}}^{(t)}_m$'s to the root node $r$. But now, their local estimates are summed up together at the node common bus $s_i$ considering the weights $g(s_i,l_m)$ -- the sum of the weights is equal to 1. This is followed by a weighted final sum between the sub-array nodes $s_i$ so as to send the final information to the root node $r$. The benefit of tree-based architectures is that they allow more processing to be performed in parallel instead of the full sequential mode provided by the daisy-chain. This comes with the expectation of further reducing the delay to output the final estimate $\hat{\mathbf{x}}$. Future work can better characterize the operation and the trade-offs of more complex tree structures, comparing them with the daisy-chain of Fig. \ref{fig:daisy-chain}.
\begin{figure}[htbp!]
    \centering
    \includegraphics[trim=0 0 0.0cm 0, clip, width=0.8\columnwidth]{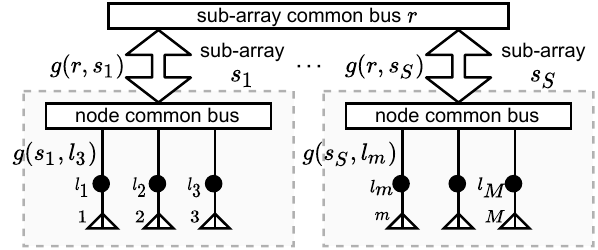}
    \caption{Decentralized baseband processing architecture at the BS including the concept of sub-arrays. \label{fig:coarse-tree}}
\end{figure}

\section{Numerical Results}\label{sec:num-res}
We now present some simulation results showing the efficiency of the BDK receiver and our method to choose the relaxation parameter for the SDK receiver. We use the bit-error rate (BER) as a measure of performance and drawn the data signal vector $\mathbf{x}$ according to a normalized 16-QAM constellation with natural binary-coded ordering.

\subsection{Stationary Case}
The stationary case comprises of all antennas serving all UEs, that is, $\mathbf{D}_m=\mathbf{I}_K, \ \forall m\in\mathcal{M}$. Fig. \ref{fig:mmimo:ber-vs-snr} shows the average BER per UE as a function of the SNR in dB. The BDK receiver has a slightly better performance than the SDK receiver, especially in low SNR where noise knowledge is critical. However, the convergence rate is far from good and an appropriate choice of the relaxation parameter $\lambda$ seems to be very important to accelerate the convergence of the SDK receiver. Our heuristic method to select $\lambda$ defined in \eqref{eq:lambda:sdk:proposed} is better than the one proposed in \cite{Sanchez2020}, just being worse for high SNR values. This is expected since in the case that the SNR is high, the expression in \eqref{eq:lambda:sdk:proposed} typically evaluates to $\lambda=1$ because of the minimum function.\footnote{One way to deal with this problem is to further limit the range $0<\lambda<1/2$, which was observed to work better in high SNR. However, convergence is naturally improved in the case of high SNR after a few cycles because the noise error does not affect the iterations much.} However, the advantage of our method can be better seen in Fig. \ref{fig:mmimo:ber-vs-cycles} that shows the average BER per UE as a function of the total number of cycles $T$. The figure uncovers the main drawback of the previous method of choosing $\lambda$ proposed in \cite{Sanchez2020}: after $T=1$ cycles the SINR expression changes, which in turn changes the optimal $\lambda$ value; hence, by increasing $T$, the constant relaxation parameter chosen before is no longer optimal and it scales future estimates erroneously. Meanwhile, our method provides improved performance until $T=4$ iterations; for $ T>4$ the error propagation dominates, what is known as the semi-convergence property of the Kaczmarz algorithm \cite{Elfving2014}.

\begin{figure}[thp]
    \centering
    \subfloat[\textit{vs} SNR in dB with $T=1$.]{\includegraphics[width=0.49\columnwidth]{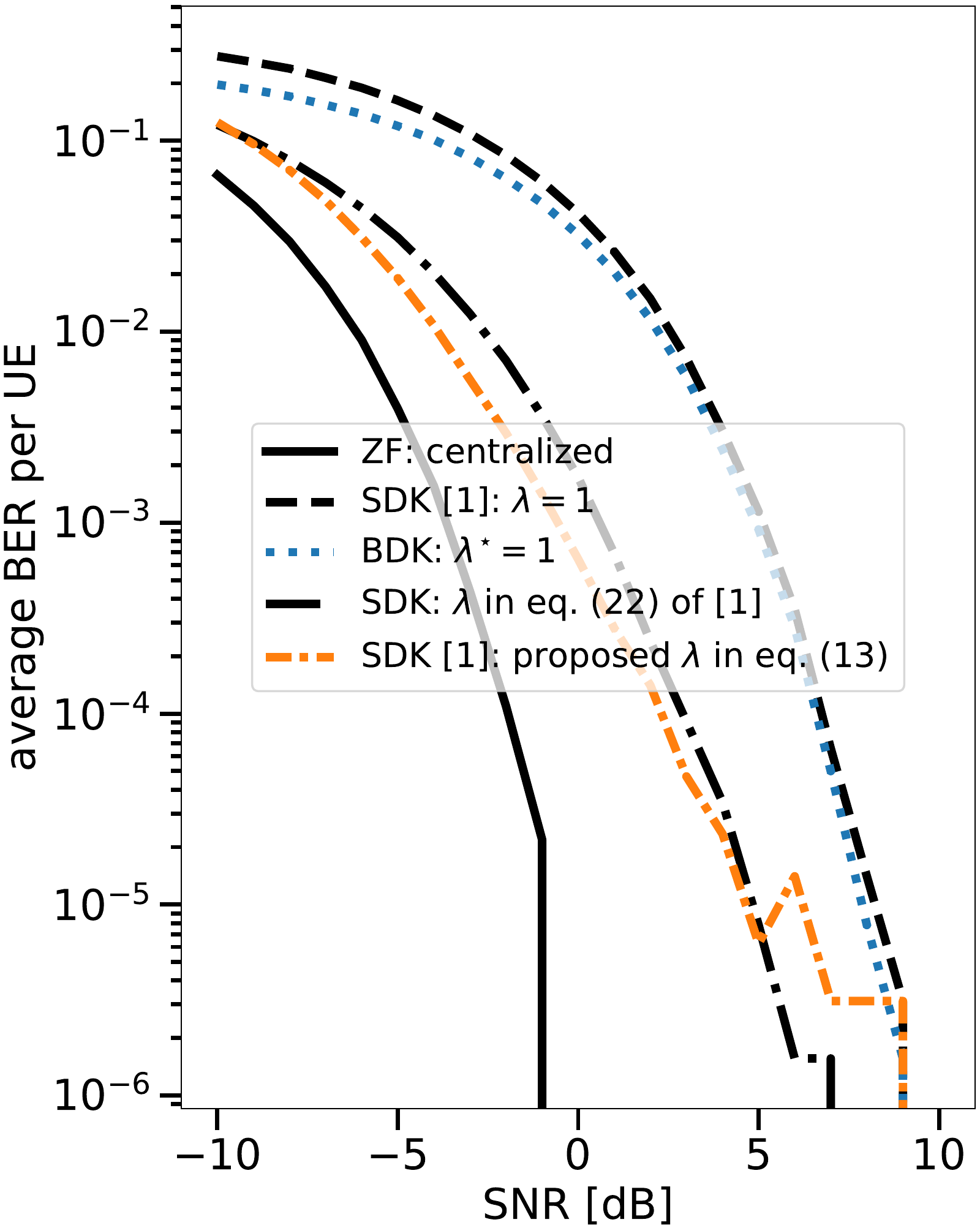}\label{fig:mmimo:ber-vs-snr}}%
    \,\hspace{-2mm}
    \subfloat[\textit{vs} $T$ with $\mathrm{SNR}=0$ dB.]{\includegraphics[width=0.49\columnwidth]{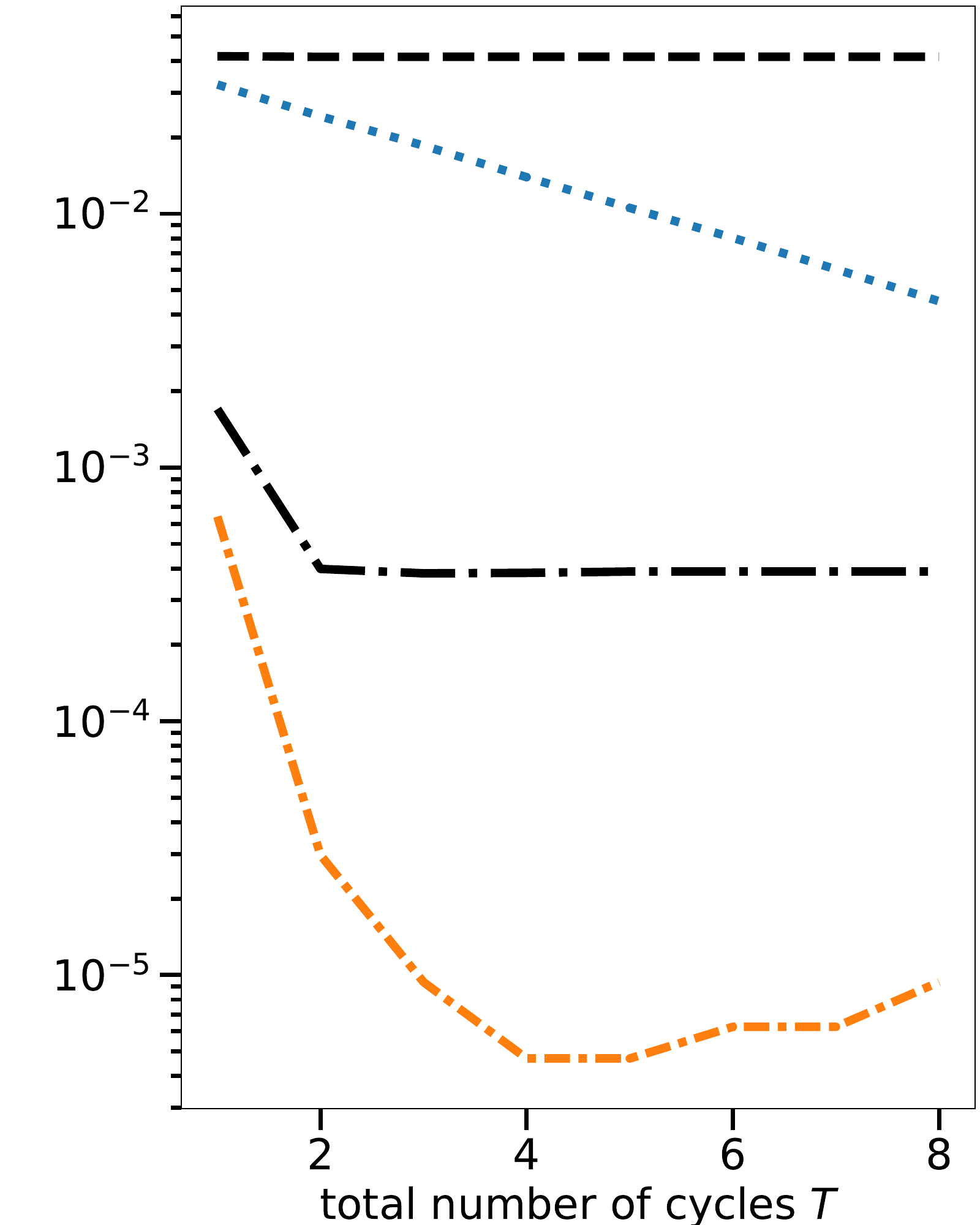}\label{fig:mmimo:ber-vs-cycles}}
    \caption{Average BER per UE with $M=128$ antennas, $K=16$ UEs.}%
\end{figure}

\subsection{Non-Stationary Case}
To generate the diagonals of the matrix $\mathbf{D}_m$, we adopted a simple model defined by the following step: for each antenna $m\in\mathcal{M}$, we successively generate a random sequence of length $K$ in which each random variable follows a Bernoulli distribution with equal probability until the sum of the random sequence is equal to the number of effective users $D$. The average BER per UE as a function of $D$ is shown in Fig. \ref{fig:xlmimo:ber-vs-D}. For very small values of $D$, we note that the SRC-based methods perform bad due to the strong sparsity condition imposed by the spatial non-stationarities that mathematically hinders the solution of the SLE. For larger values of $D$, the SDK receiver with a suitable choice of $\lambda$ starts to perform better in general and improves as $D$ comes closer to $K=16$ UEs. This happens because increasing $D$ the rank-deficiency of the channel matrix $\mathbf{H}$ is reduced; when $D=K$, $\mathbf{H}$ is full-rank \cite{Ali2019} and ZF performs better even with increasing interference. In Fig. \ref{fig:xlmimo:ber-vs-cycles}, we can note that our method of choosing $\lambda$ works better than the previous method under sparse $\mathbf{h}_m$. However, the curve converges to a floor because the method generates a scaled final estimate $\boldsymbol{\alpha}\hat{\mathbf{x}}$, where $\boldsymbol{\alpha}\in\mathbb{C}^{K\times1}\neq{[1, 1, \dots, 1]}$ embraces two facts introduced by the sparse $\mathbf{h}_m$: a) the $m$-th estimate $\hat{\mathbf{x}}_m$ does not necessarily suffer from $m-1$ independent noise sources, going against (O.2); b) the Kaczmarz iterative step non-coherently combines the estimates from one iteration to another, and it converges to a \emph{weighted} minimum norm LS solution (see Remark 1).

\begin{figure}[thp]
    \centering
    \subfloat[\textit{vs} $D$ with $\mathrm{SNR}=0$ dB, $T=1$.]{\includegraphics[width=0.49\columnwidth]{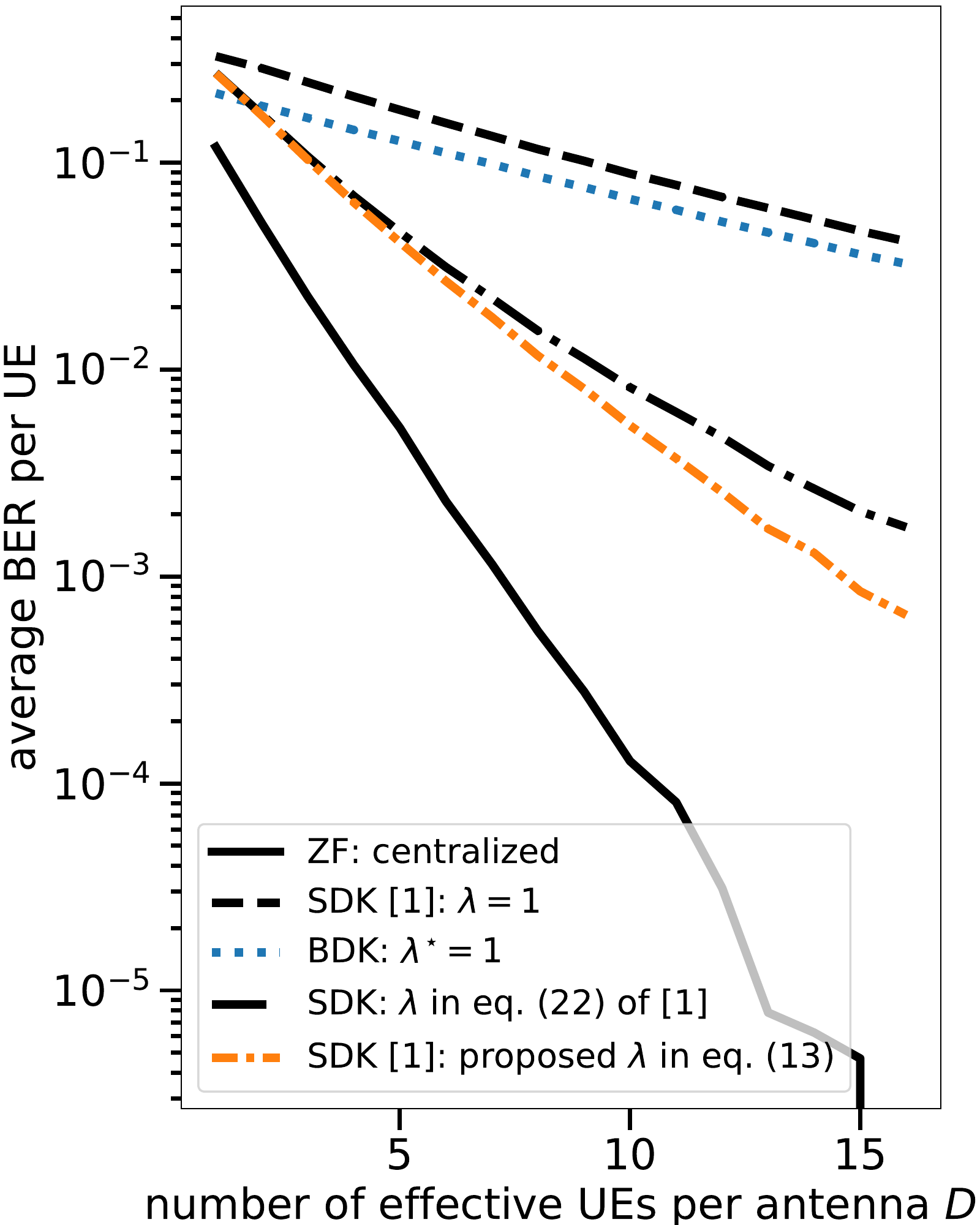}\label{fig:xlmimo:ber-vs-D}}%
    \,\hspace{-2mm}
    \subfloat[\textit{vs} $T$ with $\mathrm{SNR}=0$ dB, $D=8$.]{\includegraphics[width=0.49\columnwidth]{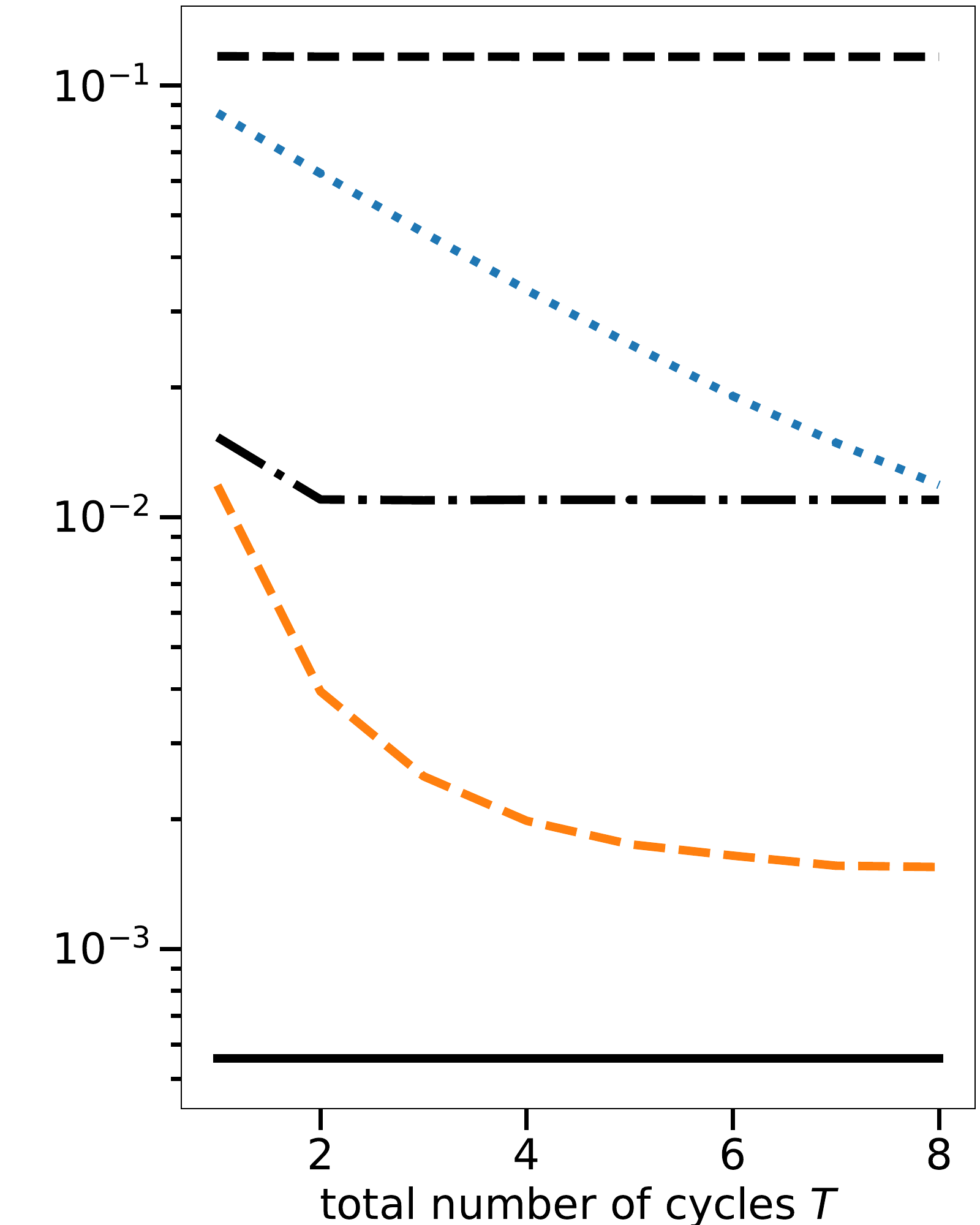}\label{fig:xlmimo:ber-vs-cycles}}
    \caption{Average BER per UE with $M=128$ antennas, $K=16$ UEs.}%
\end{figure}

\section{Conclusions}\label{sec:conclusion}
In this paper, we have shown two new approaches to derive the GD receiver proposed by \cite{Sanchez2020} based on a distributed version of the Kaczmarz algorithm and the principles of SRC. This part of the work was motivated by trying to better understand the GD receiver from \cite{Sanchez2020} and to draw new connections to help improve its initial convergence. Based on the connections made, we have proposed a new Bayesian distributed receiver and a new method to choose the relaxation parameter $\lambda$. Efforts were dedicated to improve the functionality of the discussed receivers under both spatially stationary and non-stationary channels.



\bibliographystyle{IEEEtran}
\bibliography{main.bbl}

\begin{thebibliography}{10}
\providecommand{\url}[1]{#1}
\csname url@samestyle\endcsname
\providecommand{\newblock}{\relax}
\providecommand{\bibinfo}[2]{#2}
\providecommand{\BIBentrySTDinterwordspacing}{\spaceskip=0pt\relax}
\providecommand{\BIBentryALTinterwordstretchfactor}{4}
\providecommand{\BIBentryALTinterwordspacing}{\spaceskip=\fontdimen2\font plus
\BIBentryALTinterwordstretchfactor\fontdimen3\font minus
  \fontdimen4\font\relax}
\providecommand{\BIBforeignlanguage}[2]{{%
\expandafter\ifx\csname l@#1\endcsname\relax
\typeout{** WARNING: IEEEtran.bst: No hyphenation pattern has been}%
\typeout{** loaded for the language `#1'. Using the pattern for}%
\typeout{** the default language instead.}%
\else
\language=\csname l@#1\endcsname
\fi
#2}}
\providecommand{\BIBdecl}{\relax}
\BIBdecl

\bibitem{Sanchez2020}
J.~{Rodríguez Sánchez}, F.~{Rusek}, O.~{Edfors}, M.~{Sarajlić}, and
  L.~{Liu}, ``{Decentralized Massive MIMO Processing Exploring Daisy-Chain
  Architecture and Recursive Algorithms},'' \emph{IEEE Transactions on Signal
  Processing}, vol.~68, pp. 687--700, 2020.

\bibitem{Shepard2012}
C.~Shepard, H.~Yu, N.~Anand, E.~Li, T.~Marzetta, R.~Yang, and L.~Zhong,
  ``{Argos: Practical Many-Antenna Base Stations},'' in \emph{Proceedings of
  the 18th Annual International Conference on Mobile Computing and Networking},
  ser. Mobicom '12.\hskip 1em plus 0.5em minus 0.4em\relax New York, NY, USA:
  Association for Computing Machinery, 2012, p. 53–64.

\bibitem{Carvalho2020}
E.~D. {Carvalho}, A.~{Ali}, A.~{Amiri}, M.~{Angjelichinoski}, and R.~W.
  {Heath}, ``{Non-Stationarities in Extra-Large-Scale Massive MIMO},''
  \emph{IEEE Wireless Communications}, vol.~27, no.~4, pp. 74--80, 2020.

\bibitem{Li2019}
K.~Li, J.~McNaney, C.~Tarver, O.~Castañeda, C.~Jeon, J.~R. Cavallaro, and
  C.~Studer, ``{Design Trade-offs for Decentralized Baseband Processing in
  Massive MU-MIMO Systems},'' in \emph{2019 53rd Asilomar Conference on
  Signals, Systems, and Computers}, 2019, pp. 906--912.

\bibitem{Jeon2019}
C.~Jeon, K.~Li, J.~R. Cavallaro, and C.~Studer, ``{Decentralized Equalization
  With Feedforward Architectures for Massive MU-MIMO},'' \emph{IEEE
  Transactions on Signal Processing}, vol.~67, no.~17, pp. 4418--4432, 2019.

\bibitem{Li2017}
K.~Li, R.~R. Sharan, Y.~Chen, T.~Goldstein, J.~R. Cavallaro, and C.~Studer,
  ``{Decentralized Baseband Processing for Massive MU-MIMO Systems},''
  \emph{IEEE Journal on Emerging and Selected Topics in Circuits and Systems},
  vol.~7, no.~4, pp. 491--507, 2017.

\bibitem{Amiri2021}
A.~Amiri, C.~N. Manch'on, and E.~de~Carvalho, ``{Uncoordinated and
  Decentralized Processing in Extra-Large MIMO Arrays},'' 2021.

\bibitem{Hegde2019}
C.~Hegde, F.~Keinert, and E.~S. Weber, ``{A Kaczmarz Algorithm for Solving Tree
  Based Distributed Systems of Equations},'' 2019.

\bibitem{Verdu1998}
S.~Verdú, \emph{{Multiuser Detection}}, 1st~ed.\hskip 1em plus 0.5em minus
  0.4em\relax USA: Cambridge University Press, 1998.

\bibitem{Carvalho2012}
T.~Brown, E.~{De Carvalho}, and P.~Kyritsi,
  \emph{\BIBforeignlanguage{English}{{Practical Guide to the MIMO Radio Channel
  with MATLAB Examples}}}, 1st~ed.\hskip 1em plus 0.5em minus 0.4em\relax
  Wiley, 2012.

\bibitem{Bjornson2017c}
E.~Bj{\"{o}}rnson, J.~Hoydis, and L.~Sanguinetti, ``{Massive MIMO Networks:
  Spectral, Energy, and Hardware Efficiency},'' \emph{Foundations and
  Trends{\textregistered} in Signal Processing}, vol.~11, no. 3-4, pp.
  154--655, 2017.

\bibitem{Ali2019}
A.~{Ali}, E.~D. {Carvalho}, and R.~W. {Heath}, ``{Linear Receivers in
  Non-Stationary Massive MIMO Channels With Visibility Regions},'' \emph{IEEE
  Wireless Communications Letters}, vol.~8, no.~3, pp. 885--888, 2019.

\bibitem{Kaczmarz1937}
S.~Kaczmarz, ``{Angen{\"{a}}herte Aufl{\"{o}}sung von Systemen linearer
  Gleichungen},'' \emph{Bulletin International de l'Acad{\'{e}}mie Polonaise
  des Sciences et des Lettres. Classe des Sciences Math{\'{e}}matiques et
  Naturelles. S{\'{e}}rie A, Sciences Math{\'{e}}matiques}, vol.~35, pp.
  355--357, 1937.

\bibitem{Censor1983}
Y.~Censor, P.~P.~B. Eggermont, and D.~Gordon, ``{Strong Underrelaxation in
  Kaczmarz's Method for Inconsistent Systems},'' \emph{Numerische Mathematik},
  vol.~41, no.~1, pp. 83--92, {Feb} 1983.

\bibitem{Herman2009}
G.~T. Herman, \emph{{Fundamentals of Computerized Tomography}}, ser. Advances
  in Pattern Recognition.\hskip 1em plus 0.5em minus 0.4em\relax London:
  Springer London, 2009.

\bibitem{Strohmer2009}
T.~Strohmer and R.~Vershynin, ``{Comments on the Randomized Kaczmarz Method},''
  \emph{Journal of Fourier Analysis and Applications}, vol.~15, no.~4, pp.
  437--440, 2009.

\bibitem{Bentrcia2019}
A.~Bentrcia, ``{Kaczmarz Successive Interference Cancellation: A Matrix
  Algebraic Approach},'' \emph{Digital Signal Processing}, vol.~89, pp. 60--69,
  2019.

\bibitem{Shai2014}
S.~Shalev-Shwartz and S.~Ben-David, \emph{{Understanding Machine Learning -
  From Theory to Algorithms}}.\hskip 1em plus 0.5em minus 0.4em\relax Cambridge
  University Press, 2014.

\bibitem{Hanke1990}
M.~Hanke and W.~Niethammer, ``{On the Acceleration of Kaczmarz's Method for
  Inconsistent Linear Systems},'' \emph{Linear Algebra and its Applications},
  vol. 130, pp. 83--98, 1990.

\bibitem{Elfving2014}
T.~Elfving, P.~C. Hansen, and T.~Nikazad, ``{Semi-Convergence Properties of
  Kaczmarz's Method},'' \emph{Inverse Problems}, vol.~30, no.~5, p. 055007,
  {May} 2014.

\end{thebibliography}

\appendix[Proof of Theorem \ref{theorem:01}]
Let $\mathbf{x}$ be the solution of \eqref{eq:unconstrained-least-squares} and $\boldsymbol{\varrho}^{(t)}_m=\mathbf{h}^{*}_m{r}^{(t)}_m$. We then re-define the Kaczmarz iterative step in \eqref{eq:kaczmarz-iterative-step} as $\hat{\mathbf{x}}^{(t)}_m=\hat{\mathbf{x}}^{(t)}_{m-1}+\lambda_m\boldsymbol{\varrho}^{(t)}_m$. By defining the error at the Kaczmarz iterative step applied to the $m$-th node at the first cycle\footnote{It is straightforward to generalized Theorem \ref{theorem:01} for other number of cycles, one just needs to use the solution from the previous cycle as the initialization of the next $\hat{\mathbf{x}}^{(t+1)}_0=\hat{\mathbf{x}}^{(t)}_M$.} ($t=1$) results: 
\begin{IEEEeqnarray}{rCl}
    {\hat{\mathbf{x}}^{(1)}_m-\mathbf{x}}&=&(\hat{\mathbf{x}}^{(1)}_{m-1}-\mathbf{x}) + \lambda_m\boldsymbol{\varrho}^{(1)}_m\IEEEnonumber\\
    \ltwonorm{{\hat{\mathbf{x}}^{(1)}_m-\mathbf{x}}}&=&\ltwonorm{\hat{\mathbf{x}}^{(1)}_{m-1}-\mathbf{x}}+2\lambda_m\mathcal{R}(\langle\hat{\mathbf{x}}^{(1)}_{m-1}-\mathbf{x},\boldsymbol{\varrho}^{(1)}_m\rangle)\IEEEnonumber\\
    &+&\lambda_m^2\ltwonorm{\boldsymbol{\varrho}^{(1)}_m}\IEEEnonumber\\
    \mathcal{R}(\langle\hat{\mathbf{x}}^{(1)}_{m-1}-\mathbf{x},\boldsymbol{\varrho}^{(1)}_m\rangle)&=&\dfrac{1}{2\lambda_m}(\ltwonorm{{\hat{\mathbf{x}}^{(1)}_m-\mathbf{x}}}-\ltwonorm{{\hat{\mathbf{x}}^{(1)}_{m-1}-\mathbf{x}}})\IEEEnonumber\\ 
    &-&\dfrac{\lambda_m}{2}\ltwonorm{\boldsymbol{\varrho}^{(1)}_m}\IEEEnonumber.
\end{IEEEeqnarray}
Summing the above equality over $m\in\mathcal{M}$ yields in
\begin{IEEEeqnarray}{rCl}
    \sum_{m=1}^{M}\mathcal{R}(\langle\hat{\mathbf{x}}^{(1)}_{m-1}-\mathbf{x},\boldsymbol{\varrho}^{(1)}_m\rangle)&=&\sum_{m=1}^{M}\dfrac{1}{2\lambda_m}(\ltwonorm{{\hat{\mathbf{x}}^{(1)}_m-\mathbf{x}}}\IEEEnonumber\\
    &-&\ltwonorm{{\hat{\mathbf{x}}^{(1)}_{m-1}-\mathbf{x}}}) - \sum_{m=1}^{M}\dfrac{\lambda_m}{2}\ltwonorm{\boldsymbol{\varrho}^{(1)}_m}\IEEEnonumber\\
    &\stackrel{\text{(a)}}{=}&\dfrac{1}{2\lambda_M}\ltwonorm{{\hat{\mathbf{x}}^{(1)}_M-\mathbf{x}}}-\dfrac{1}{2\lambda_1}\ltwonorm{{\hat{\mathbf{x}}^{(1)}_0-\mathbf{x}}}\IEEEnonumber\\ 
    &-&\sum_{m=1}^{M}\dfrac{\lambda_m}{2}\ltwonorm{\boldsymbol{\varrho}^{(1)}_m}\IEEEnonumber\\ 
    &\stackrel{\text{(b)}}{\geq}&-\dfrac{1}{2\lambda_1}\ltwonorm{\mathbf{x}}-\sum_{m=1}^{M}\dfrac{\lambda_m}{2}\ltwonorm{\boldsymbol{\varrho}^{(1)}_m}\IEEEnonumber, 
\end{IEEEeqnarray}
where in (a) we use the fact that the first sum on the right-hand side is a telescopic sum, and in (b) we have lower bounded the left-hand side and assumed that $\hat{\mathbf{x}}^{(1)}_0=\mathbf{0}_{K\times1}$. Eq. \eqref{eq:theorem:sse:02} then follows by taking the expectation of the above expression with respect to the data signal vector $\mathbf{x}$ and the noise vector $\mathbf{n}$ conditioned on the perfect CSI knowledge denoted by the set $\{\mathbf{h}_m\}, \ \forall m\in\mathcal{M}$. Then, lower bounding $\mathbb{E}[\ltwonorm{\mathbf{x}}|\{\mathbf{h}_m\}]$ by $B^2$, $\mathbb{E}[\ltwonorm{\boldsymbol{\varrho}^{(t)}_m}|\{\mathbf{h}_m\}]$ by $\rho^2$ for every $m\in\mathcal{M}$, and defining $\lambda_m$ to incorporate $B$, $\rho$, $\sqrt{m}$, $\sqrt{t}$ in addition to $\ltwonorm{\mathbf{h}_{m}}$, which comes from the definition of the Kaczmarz iterative step in \eqref{eq:kaczmarz-iterative-step}, leads to \eqref{eq:theorem:sse:03}. It is straightforward to obtain the loose lower bound $B^2=Kp$ from $\mathbf{x}\sim\mathcal{CN}(\mathbf{0},p\mathbf{I}_K)$. For $\rho$, we consider the worst-case that the residual in $\boldsymbol{\varrho}^{(t)}_m$ is pure noise and has the form $n_m\sim\mathcal{CN}(0,\sigma^2)$, then $\rho^2=\sigma^2$; the channel is normalized out to set $B^2$ and $\rho^2$ on the same power scale. 


\end{document}